\documentclass{llncs} 

\usepackage{amssymb}
\usepackage{textcomp}
\usepackage{amsmath}
\usepackage{algorithm}
\usepackage{algpseudocode}
\usepackage{comment}
\usepackage{graphicx}
\usepackage{color}
\newcommand{\Endproof}{\hfill$\Box$\\}

\DeclareMathOperator{\dyck}{\textsc{Dyck}}
\newcommand{\sign}{\operatorname{sign}}
\newcommand{\argmax}{\operatorname{argmax}}

\newcommand{\true}{\textsc{True}}
\newcommand{\false}{\textsc{False}}

\newcommand{\ket}[1]{|#1\rangle}

\begin{document}

\title{Quantum algorithm for Dyck Language with Multiple Types of Brackets}
\author{Kamil Khadiev\inst{1,2}, Dmitry Kravchenko\inst{3}
} 

\institute{Kazan Federal University, Kazan, Russia, \and
Kazan E. K. Zavoisky Physical-Technical Institute, Kazan, Russia
\and
Center for Quantum Computer Science, Faculty of Computing, University of Latvia, Riga, Latvia\\
 \email{kamilhadi@gmail.com, kravchenko@gmail.com} }

\maketitle

\begin{abstract}
We consider the recognition problem of the Dyck Language generalized for multiple types of brackets. We provide an algorithm with quantum query complexity $O(\sqrt{n}(\log n)^{0.5k})$, where $n$ is the length of input and $k$ is the maximal nesting depth of brackets. Additionally, we show the lower bound for this problem which is $O(\sqrt{n}c^{k})$ for some constant $c$.

Interestingly, classical algorithms solving the Dyck Language for multiple types of brackets substantially differ form the algorithm solving the original Dyck language. At the same time, quantum algorithms for solving both kinds of the Dyck language are of similar nature and requirements.

\end{abstract}

\section{Introduction}
\emph{Quantum computing} \cite{nc2010,a2017,aazksw2019part1} is one of the hot topics in computer science of the last decades.
There are many problems where quantum algorithms outperform the best known classical ones \cite{quantumzoo},
and one of the most important performance metrics in this regard is \emph{query complexity}.
We refer to \cite{a2017} for a nice survey on the quantum query complexity,
and to \cite{ks2019,kks2019,kksk2020,kms2019,gnbk2021} for the more recent progress.

Among other problems, quantum technologies can reduce the query complexity of recognizing many formal languages 
.
In this paper we consider a problem of recognizing whether an $n$-bit string belongs to one important \emph{regular language}.
Although this problem may seem too specific, we believe our approach to model a variety of computational tasks that can be described by regular languages.

Aaronson, Grier and Schaeffer \cite{ags2019} have recently shown that any regular language $L$ may have one of three possible quantum query complexities on inputs of length $n$: 
$\Theta(1)$ if the language can be decided by looking at $O(1)$ first or last symbols of a word;
$\tilde{\Theta}(\sqrt{n})$ if the best way to decide $L$ is Grover's search (for example, for the language consisting of all words containing at least one letter $a$);
$\Theta(n)$ for languages in which one can embed counting modulo some number $p$ which has quantum query complexity $\Theta(n)$ (for example, the binary $XOR$ function).

As shown in \cite{ags2019}, a regular language being of complexity $\tilde{O}(\sqrt{n})$ (which includes the first two cases of the list above) is equivalent to it being \emph{star-free}.
Star-free languages are defined as the languages which have regular expressions not containing the Kleene star (if it is allowed to use the complement operation).
Star-free languages are one of the most commonly studied subclasses of regular languages and there are many equivalent characterizations of them.

One class of the star-free languages mentioned in \cite{ags2019} is the Dyck languages (with one type of brackets and with constant height $k$).
To introduce a brief intuition about these languages, we may mention that words ``\texttt{[ ]}'' and ``\texttt{[~[~]~[~]~]}'' belong to a Dyck language, while words ``\texttt{]~[}'' and ``\texttt{[~]~]~[~[~]}'' do not.
Formally, Dyck language with height $k$ consists of all words with balanced number of brackets such that in no prefix the number of opening brackets exceeds the number of closing brackets by more than $k$;
we denote the problem of determining if an input of length $n$ belongs to this language by $\dyck_{k,n}$.
We note that such language of unbounded height (i.e. $k=\frac{n}{2}$) is a fundamental example of a context-free language that is not regular.
%

For this problem, Ambainis et al. \cite{abikkpssv2020} show that an exponential dependence of the complexity on $k$ is unavoidable. Namely, for the balanced brackets language
(i) there exists $c>1$ such that, for all $k\leq \log n$, the quantum query complexity is $\Omega(c^k \sqrt{n})$;
(ii) if $k=c\log n$ for an appropriate constant $c$, then the quantum query complexity is $\Omega(n^{1-\epsilon})$. 

Thus, the exponential dependence on $k$ is unavoidable and distinguishing sequences of balanced brackets of length $n$ and depth $\log n$
is almost as hard as distinguishing sequences of length $n$ and arbitrary depth.
Similar lower bounds have recently been independently proven by Buhrman et al. \cite{bps2019}.
Additionally, Ambainis et al. \cite{abikkpssv2020} describe an explicit algorithm for the decision problem $\dyck_{k,n}$ with $O\left(\sqrt{n}(\log{n})^{0.5k}\right)$ quantum queries.
The algorithm also works for arbitrary $k$, and outperforms the trivial upper bound of $n$ when $k=o\left(\frac{\log n}{\log\log n}\right)$. 

This work generalizes $\dyck_{k,n}$ to the case of multiple types of brackets.
For example, such languages contain words like ``\texttt{[~(~)~]}'' and do not contain words like ``\texttt{[~(~]~)}'' (here square and round brackets are the two different types of brackets).
We denote the problem of determining if an input of length $n$ belongs to the Dyck language of height $k$ and at most $t$ types of brackets by $\dyck_{k,n,t}$.
Obviously, $\dyck_{k,n,1}=\dyck_{k,n}$.

We note that $\dyck_{k,n}$ and $\dyck_{k,n,t}$ for $t>1$ are two substantially different problems regarding classical (deterministic or randomized) calculations.
The former problem allows using a counter to keep the number of currently open brackets and thus be content with the memory size of $O(\log k)$.
In contrast, the latter problem requires keeping all the sequence of currently open brackets in a stack, which may take up to $O(k)$ memory.
While both problems are solvable in linear time, there is an exponential gap in the memory usage.

In this paper we provide a quantum algorithm for $\dyck_{k,n,t}$ with quantum query complexity $O(\sqrt{n}(\log n)^{0.5k})$.
We apply the known technique of solving $\dyck_{k,n}$, and then perform a more complex but slightly faster procedure to check the type-matching of the brackets.

The structure of the paper is the following.
Section \ref{sec:prelims} describes some conventional notions for quantum computation.
Section \ref{sec:algo} provides the main algorithm and the proofs.
The final subsection \ref{sec:algo-compl} contains the discussion on the complexity of the algorithm and on the lower bounds.

\section{Preliminaries}\label{sec:prelims}

\subsection{Definitions}

We use the following formalism throughout the paper.
We assume an input string to consist of brackets of $t$ \emph{types} for some positive integer $t$; each type is represented by a pair of brackets -- an opening and a closing one.
Further, we assume the brackets to be encoded by integers from $1$ to $2t$, where the opening and the closing brackets of $i$-th type correspond to the numbers $2i-1$ and $2i$ respectively.

We define two functions:
\begin{itemize}
\item Function $Type:\{1,\dots,2t\}\to\{1,\dots,t\}$ returns the type of a bracket. \\
      $Type(x)=\lceil x/2\rceil$.
\item Function $Open:\{1,\dots,2t\}\to\{0,1\}$ returns $1$ if the argument is an opening bracket, or $0$ if it is a closing bracket.\\
      $Open(x) = x \mod 2$.
\end{itemize}
\begin{flushleft}
For example, string ``\texttt{[~(~)~]}'' could be encoded as ``$1,3,4,2$''. Then \\
$Type(1)=Type(4)=1$ \quad stand for the square brackets; \\
$Type(2)=Type(3)=2$ \quad -- for the parentheses; \\
$Open(1)=Open(2)=1$ \quad -- for the opening brackets; and \\
$Open(3)=Open(4)=0$ \quad -- for the closing brackets.
\end{flushleft} 

We call a string $S=(s_1,\dots,s_m)$ a \emph{well-balanced} sequence of brackets if one of the following holds:
\begin{enumerate}
\item $S$ is empty;
\item $S$ consists of two well-balanced subsequent substrings, i.e. $S[1,i]$ and $S[i+1,m]$ are both well-balanced for some $i$
    (hereafter we denote by $S[i,j]$ a substring $(s_i,\dots,s_j)$ of a string $S=(s_1,\dots,s_m)$);
\item $S$ is a correctly bracketed well-balanced sequence, i.e.
    \begin{itemize}
    \item $S[2,m-1]$ is a well-balanced sequence,
    \item $Type(s_1)=Type(s_m)$,
    \item $Open(s_1)=1$ and $Open(s_m)=0$.
    \end{itemize}
\end{enumerate}
Obviously, the set of all well-balanced sequences of brackets defines the $\dyck$ language.

We also introduce a metric for the balancedness of a substring.
Let $f$ be a function which returns the difference between the numbers of opening and closing brackets:
$f(S[l,r]) = \#_1(S[l,r])-\#_0(S[l,r]).$ (Here $\#_x(S[l,r])$ denotes the number of symbols $s_j$, for $l\leq j\leq r$, such that $Open(s_j)=x$).
We define a $+k$-substring (resp. $-k$-substring) as a substring whose balance is equal to $k$ (resp. equal to $-k$).
A $\pm k-$substring is a substring whose balance is equal to $k$ in absolute value.

We call a nonempty substring $S[l,r]$ \emph{minimal} if it does not contain a nonempty substring $S[l',r']$ such that $(l,r)\neq (l',r')$ and $f(S[l',r'])=f(S[l,r])$.
We call a nonempty substring $S[l,r]$ \emph{prefix-minimal} if it does not start with $S[l,r']$ such that $r'< r$ and $f(S[l,r'])=f(S[l,r])$.
We define the \emph{height} of a substring $S[l,r]$ as $h(S[l,r])=\max_{i\in\{l,\ldots,r\}}f(S[l,i])$.

For example, string $S=``\texttt{[~]~(~)}''$ is well-balanced, because it consists of two well-balanced substrings ``\texttt{[~]}'' and ``\texttt{(~)}'', which in turn both are correctly embraced empty strings.
Its substring $S[1,2]=``\texttt{[~]}''$ is both minimal and prefix-minimal, whereas its substring $S[2,4]=``\texttt{]~(~)}''$ is neither minimal nor prefix-minimal (since $f(S[2,2])=f(S[2,4])=-1$).

Finally, we define the problem $\dyck_{k,n,t}(S)$.
Function $\dyck_{k,n,t}$ accepts $S=(s_1,\dots,s_n)$ as an input and
\begin{itemize}
\item returns $1$ if $S$ is a well-balanced sequence of brackets with at most $t$ types of brackets and with $h(S)\leq k$;
\item returns $0$ otherwise.
\end{itemize}

\subsection{Computational Model}
To evaluate the complexity of a quantum algorithm, we use the standard form of the quantum query model.
It is a generalization of the decision tree model of classical computation that is commonly used to lower bound the amount of time required for a computation.

Let $f:D\rightarrow \{0,1\}$, for some $D\subseteq \{0,1\}^n$, be an $n$-argument binary function we wish to compute.
We have an oracle access to the input $x$ --- it is implemented by a specific unitary transformation usually defined as
\mbox{$\ket{i}\ket{z}\ket{w}\rightarrow \ket{i}\ket{z \oplus x_i}\ket{w}$},
where the $\ket{i}$ register indicates the index of the variable we are querying,
$\ket{z}$ is the output register, and $\ket{w}$ is some auxiliary work-space.
An algorithm in the query model consists of alternating applications of arbitrary unitaries independent of the input and the query unitary, and a measurement in the end.
The smallest number of queries for an algorithm that outputs $f(x)$ with probability $\geq \frac{2}{3}$ on all $x$ is called the quantum query complexity of the function $f$ and is denoted by $Q(f)$.
Throughout this paper, by the running time of an algorithm we mean a number of queries to oracle.

In particular, in this paper we assume the oracle to process queries $Type$ and $Open$ in constant time.

More information on quantum computation and query model can be found in \cite{nc2010,a2017,aazksw2019part1}.

To distinguish ordinary deterministic and randomized complexities from the quantum complexity, they are traditionally called by one term \emph{classical complexity}.
    
\section{Quantum Algorithm}\label{sec:algo}

Before introducing the algorithm for solving $\dyck_{k,n,t}(S)$, we mention the following result from \cite{abikkpssv2020}, which will be used as important subroutine.

\begin{lemma}[\cite{abikkpssv2020}, Theorem 3]\label{lm:dyck1}
There exists a quantum algorithm that solves $\dyck_{k,n,1}$ in time $O(\sqrt{n}(\log n)^{0.5k})$. The algorithm has two-side bounded error probability $\varepsilon<0.5$.
\end{lemma}


The algorithm for solving $\dyck_{k,n,t}(S)$ generally consists of three main steps:

\textbf{Step 1.} Check whether there are at most $t$ types of brackets, and return $0$ if the number of types exceeds $t$.
                 This part is discussed in Section \ref{sec:t-types}.

\textbf{Step 2.} Uniformize $S$ to just one type of brackets by considering a string $Y=(y_1,\dots,y_n)$ where $y_i=Open(s_i)$.
                 Check whether $\dyck_{k,n,1}(Y)=1$ by using the algorithm from Lemma \ref{lm:dyck1}.
                 If this is the case, then $S$ is a well-balanced sequence of brackets with their types ignored.
                 Otherwise, $S$ obviously is not well-balanced and $\dyck_{k,n,t}(S)=0$. This step almost exactly repeats the algorithm from \cite{abikkpssv2020}.

\textbf{Step 3.} Check whether for any substring $S[l,r]$ the following condition holds: 
\noindent
                 If $Y[l,r]$ is a well-balanced sequence of brackets (with their types ignored) of depth $v$ and $Y[l+1,r-1]$ is a well-balanced sequence of brackets of depth $v-1$, then
                 (1) $Type(s_{l})=Type(s_{r})$; and
                 (2) $S[l+1,r-1]$ is a well-balanced sequence of brackets.

Step 3 should be considered as the main contribution of the paper, and we describe it in detail in Section \ref{sec:w-balance}.
By the definition of the problem, if $S$ passes all three checks, then $\dyck_{k,n,t}(S)=1$.
The complexity of the problem is evaluated in Section \ref{sec:algo-compl}

\subsection{The Procedure for Step 1}\label{sec:t-types}
Recall that by the assumption, all the brackets are encoded by integers from $1$ to $2t$.
Hence it only remains to check whether $S$ contains a bracket with code $c>2t$.
This problem obviously can be solved by Grover's algorithm \cite{g96,bbht98} for finding an argument $j$ (if any) such that $g(j)=1$, for an arbitrary function $g:\{1,\dots,n\}\to\{0,1\}$ implemented as a quantum oracle.
Grover's algorithm runs in time $O(\sqrt{n})$ and has error probability at most $0.5$.

The assumption on encoding of brackets could be relaxed by allowing to use any distinct integer for each kind of bracket.
Then the problem becomes more complex: determine whether a sequence $s_1,\ldots,s_n$ contains at most $2t$ distinct integers.
The upper bound for its query complexity is $O(\sqrt{n}t\log{t})$.
We refer to Section \ref{sec:maxsearch} for more details.


\subsection{The Procedure for Step 3}\label{sec:w-balance}
Assume that any $0$-substring $S[l',r']$ with $h(S[l',r]')\leq v-1$ is known to be a well-balanced sequence of brackets.
In this section we present a procedure that checks whether, under this assumption, any $0$-substring $S[l,r]$ with $h(S[l,r])=v$ is a well-balanced sequence of brackets.

We wish to implement a function $\textsc{CheckSubstr}(S,v)$ which returns 
\begin{itemize}
    \item $\true$ if there exists a ``wrong'' (not well-balanced) sequence $S[l,r]$ such that \sloppy{$h(S[l,r])=v$};
    \item $\false$ otherwise.
\end{itemize}

If we had $\textsc{CheckSubstr}(S,v)$ implemented, then we could invoke it for each $v\in\{1,\dots,k\}$.
In case of all-$\false$ output, the function should return $\false$ (``no wrong sequences''), otherwise $\true$ (``found a wrong sequence for at least one height $v\in\{1,\dots,k\}$'').

We propose the following implementation of $\textsc{CheckSubstr}(S,v)$.

\subsubsection{The case $v=1$.}
We start with considering the case $v=1$. Let a function $g^1 :\{1,\dots,n-1\}\to\{0,1\}$ be such that $g^1(j)=1$ iff $Open(s_j)=1$, $Open(s_{j+1})=0$, and $Type(s_j)\neq Type(s_{j+1})$. In other words, the function indicates sequentially open and close brackets of different types.

We use Grover's algorithm to search for an argument $j\in\{1,\dots,n\}$ such that $g^1(j)=1$.
Hereafter we call this subroutine $\textsc{Grover}(g^1,1,n)$, where $g^1$ is the function run by a quantum oracle in constant time, and $1\ldots n$ defines an interval to search in.
If $\textsc{Grover}(g^1,1,n)$ finds such index $j$, then $\textsc{CheckSubstr}(S,1)$ returns $\true$, otherwise $\false$.

Note that due to the complexity of Grover's algorithm, the query complexity of $\textsc{Grover}(g^1,1,n)$  is $O(\sqrt{n})$, with the error probability at most $0.5$.

\subsubsection{The case $v>1$.}
This step allows assuming any $0$-substring $S[l',r']$ with $h(S[l',r'])=v-1$ to be a well-balanced sequence of brackets. 
Under this assumption, we show that the next property holds:

\begin{lemma}
If for an input string $S$, any $0$-substring $S[l',r']$ with $h(S[l',r'])=v-1$ is a well-balanced sequence of brackets, then any prefix-minimal $0$-substring $S[l,r]$ with $h(S[l,r])=v$ is such that $S[l+1,r-1]$ is either empty or a well-balanced sequence of brackets.

\end{lemma}
\begin{proof}
According to the definition of a prefix-minimal $0$-substring, we claim that $S[l,r]$ does not contain any shorter prefix $0$-substring.
In particular, it means that $Open(s_l)=1$ and $Open(s_r)=0$. Therefore, $h(S[l+1,r-1])=v-1$, and $S[l+1,r-1]$ is a $0$-substring. Due to the claim of the lemma, $S[l+1,r-1]$ is a well-balanced sequence of brackets.
\Endproof
\end{proof}

Therefore, to complete checking whether the $0$-substring $S[l,r]$ with $h(S[l,r])=v$ is a well-balanced sequence of brackets, it only remains to check that $Type(s_l)=Type(s_r)$.

Let us present a subroutine that searches for a $0$-substring $S[l,r]$ with $h(S[l,r])=v$ such that $Type(s_l)\neq Type(s_r)$. If this subroutine finds nothing, it means that any $0$-substring $S[l,r]$ with $h(S[l,r])=v$ is well-balanced.

We use the following property of prefix-minimal $0$-substrings:

\begin{lemma}
For any prefix-minimal $0$-substring $S[l,r]$ with $h(S[l,r])=v$, there exist   indices $r'$ and $l'$ such that
\begin{itemize}
    \item $l\leq r'<l' \leq r$,
    \item $S[l,r']$ is a $+v$-substring,
    \item $S[l',r]$ is a $-v$-substring, and
    \item there are no $\pm v$-substrings contained in $S[r'+1,l'-1]$.
\end{itemize} 
\end{lemma}
\begin{proof}
Assume that there is no such index $r'\in\{l,\dots,r-1\}$ that $S[l,r']$ is a $+v$-substring. Then we consider the index $j=\argmax_{j\in\{l+1,\dots,r\}}f(S[l,j])$ and note that $h(S[l,r])=v$ implies $f(S[l,j])=v$, which contradicts the assumption.
We conclude the that the desired index $r'$ exists.

Now assume that there is no such index $l'\in\{r'+1,\dots,r\}$ that $S[l',r]$ is a $-v$-substring. Recall that by the definition of a $0$-substring, $f(S[l,r])=0$. At the same time, $f(S[l,r])=f(S[l,r'])+f(S[r'+1,r])$ and $f(S[l,r'])=v$. Therefore, $f(S[r'+1,r])=f(S[l,r])-f(S[l,r'])=0-v=-v$, which contradicts the assumption. We conclude that both desired indices $r'$ and $l'$ exist.

Finally, assume sequence $S[r'+1,l'-1]$ to contain a $\pm v$-substring. Then we consider the leftmost $\pm v$-substring $S[l'',r'']$, where $r'< l''\leq r''< l'$.

If $S[l'',r'']$ is a $+v$-substring, i.e. $f(S[l'',r''])=v$, then the minimality of $l''$ implies $f(S[r'+1,l''-1])>-v$. Then,
\begin{multline*}
f(S[l,r'']) = f(S[l,r'])+f(S[r'+1,l''-1])+f(S[l'',r'']) \\
= v+f(S[r'+1,l''-1])+f(S[l'+1,l''-1]) > v
\end{multline*}
contradicts the fact that $h(S[l,r])=\max_{j\in\{l+1,r\}}f(S[l,j])=v$.

To finish the proof, it remains only to consider (the impossibility of) the case where $S[l'',r'']$ is a $-v$-substring, i.e. $f(S[l'',r''])=-v$. In this case $f(S[r'+1,l''-1])$ can be negative, zero, or positive.
\begin{itemize}

    \item If $f(S[r'+1,l''-1])<0$, then
    \begin{multline*}
    f(S[l,r'']) = f(S[l,r'])+f(S[r'+1,l''-1])+f(S[l'',r'']) \\
    = v+f(S[r'+1,l''-1])-v = f(S[r'+1,l''-1]) < 0.
    \end{multline*}
    Therefore, there exists such index $j$ that $j<r''<r$ and $f(S[l,j])=0$,
    which contradicts the prefix-minimality of the $0$-substring $S[l,r]$.
    
    \item If $f(S[r'+1,l''-1])=0$, then
    
    $f(S[l,r''])=f(S[l,r'])+f(S[r'+1,l''-1])+f(S[l'',r''])=0.$
    
    Therefore, $f(S[l,r''])=0$ where $r''<r$, which contradicts the prefix-minimality of the $0$-substring $S[l,r]$.
    \item If $f(S[r'+1,l''-1])>0$, then
    
    $f(S[l,l''-1])=f(S[l,r'])+f(S[r'+1,l''-1])=v+f(S[r'+1,l''-1])>v$
    
    contradicts the fact that $h(S[l,r])=\max_{j\in\{l+1,r\}}f(S[l,j])=v$.
\end{itemize}
\Endproof
\end{proof}

These lemmas allow to formulate the algorithm for searching for a not well-balanced $0$-substring, with its length limited to be at most $d$:

\begin{itemize}
\item[Step 1.] Pick index $b$ uniformly at random in $\{1,\dots, n\}$.
\item[Step 2.] Search for the leftmost $\pm v$-substring with length at most $d$, in \sloppy{$S[b,\min(n,b+d-1)]$}. If such substring $S[i_r,j_r]$ was found, proceed to Step 3. Otherwise proceed to Step 4.
\item[Step 3.] Search for the rightmost $\pm v$-substring with length at most $d$ in $S[\max(i_r-d,1),i_r-1]$. If such substring $S[i_l,j_l]$ was found, proceed to Step 6. Otherwise stop and return $\false$.
\item[Step 4.] Search for the rightmost $\pm v$-substring with length at most $d$ in $S[\max(b-d+1,1),b]$. If such substring $S[i_l,j_l]$ was found, proceed to Step 5. Otherwise stop and return $\false$.
\item[Step 5.] Search for the leftmost $\pm v$-substring with length at most $d$ in $S[j_l+1,\min(n,j_l+d)]$. If such substring $S[i_l,j_l]$ was found, proceed to Step 6. Otherwise stop and return $\false$.
\item[Step 6.] If $f(S[i_l,j_l])>0$, $f(S[i_r,j_r])<0$ and $Type(i_l)\neq Type(j_r)$, then return the resulting substring $S[i_l,j_r]$. Otherwise stop and return $\false$.
 \end{itemize}

To search for the rightmost $\pm v$-substring or for the leftmost $\pm v$-substring of length at most $d$ in a segment, we use a subroutine from \cite{abikkpssv2020} with the following property:
\begin{lemma}[\cite{abikkpssv2020}, Property 2]
There is a quantum algorithm for searching for the leftmost or for the rightmost $\pm v$-substring of length at most $d$, in a substring $S[l,r]$. The query complexity of the algorithm is $O(\sqrt{r-l}(\log (r-l))^{0.5(v-2)})$. It returns $(i,j,\sigma)$ such that $S[i,j]$ is a $\pm v$-substring and $\sign(f(S[i,j]))=\sigma$. It returns $\false$ if such substring does not exist.
\end{lemma}

Hereafter we call subroutines for the leftmost and for the rightmost $\pm v$-substring respectively $\textsc{Leftmost}(S,l,r,v,d)$ and  $\textsc{Rightmost}(S,l,r,v,d)$. They return a triple $(i,j,\sigma)$, such that $S[i,j]$ is the resulting substring and $\sigma=sign(f(S[i,j]))$. They return $\false$ if there are no such $\pm v$-strings.

We formalize the algorithm in the code listing of Algorithm \ref{alg:fixLen}:
\begin{algorithm}
\caption{Search for a not well-balanced $0$-substring $S[l,r]$ with height $h(S[l,r])=v$ and length $r-l+1 \leq d$.}\label{alg:fixLen}
\begin{algorithmic}
\State $\{1,\dots,n\} \xleftarrow{R} b$\Comment{randomly pick $b$} 
\State $u_r=(i_r,j_r,\sigma_r)\gets\textsc{Leftmost}(S,b,\min(n,b+d-1),v,d)$
\If{$u_l\neq \false$}
\State $u_l=(i_l,j_l,\sigma_l)\gets\textsc{Rightmost}(S,\max(i_r-d,1),i_r-1,v,d)$
\Else
\State $u_l=(i_l,j_l,\sigma_l)\gets\textsc{Rightmost}(S,\max(b-d+1,1),b,v,d)$
\If{$u_l\neq \false$}
\State $u_r=(i_r,j_r,\sigma_r)\gets\textsc{Leftmost}(S,j_l+1,(n,j_l+d),v,d)$
\EndIf
\EndIf
\If{$u_l\neq \false$ and $u_r\neq \false$ and $\sigma_l=1$ and $\sigma_r=-1$ and $Type(s_{i_l})\neq Type(s_{j_r})$}
\State \Return $(i_l,j_r)$
\Else
\State \Return $\false$
\EndIf

\end{algorithmic}
\end{algorithm}

Assume that some string $S$ contains a not well-balanced $0$-substring $S[l,r]$ with height $h(S[l,r])=v$ and length $d$. The probability of finding such substring by this algorithm is equal to the probability of picking an index inside the substring, and therefore can be estimated by $\Omega(d/n)$. By applying the Amplitude amplification algorithm \cite{bhmt2002} for the randomized Algorithm \ref{alg:fixLen}, we obtain an algorithm with query complexity $O(\sqrt{\frac{n}{d}}\cdot{\sqrt{d}}(\log d)^{0.5(v-2)})=O(\sqrt{n}(\log d)^{0.5(v-2)})$. 

Next, we search for $d$ among the elements of set $T=\{2^0, 2^1,2^2, \dots, 2^{\lceil\log_2 n\rceil}\}$. This can be done also by using Grover's algorithm. The overall complexity of the algorithm for finding a $0$-substring $S[l,r]$ with height $h(S[l,r])=v$ and arbitrary length is $O(\sqrt{n}(\log n)^{0.5(v-1)})$. We note that Grover's algorithm relies on an oracle with a two-side bounded error, whereas it is hardly justified to assume a quantum oracle which directly handles $T$ to markup the appropriate lengths. To address this issue, we use the modification of the algorithm presented in \cite{abikkpssv2020,hmw2003} and thus obtain the implementation of $\textsc{CheckSubstr}(S,v)$.  

Finally, we implement Step 3 in the code listing of Algorithm \ref{alg:step3}.

\begin{algorithm}
\caption{$\textsc{Step3}(S)$}\label{alg:step3}
\begin{algorithmic}
\State $v\gets 1$
\While{$v\leq k$}
\If{$\textsc{CheckSubstr}(S,v)\neq \false$}
\State \Return $\true$
\EndIf
\State $v\gets v+1$
\EndWhile
\State \Return $\false$
\end{algorithmic}
\end{algorithm}

Then the overall algorithm for the problem $\dyck_{n,k,t}$ can be formalized as in the code listing of Algorithm \ref{alg:dyck-t}.



\begin{algorithm}
\caption{Solving $\dyck_{n,k,t}$}
\label{alg:dyck-t}
\begin{algorithmic}
\If{$\textsc{Step1}(S)=1$ \textbf{and} $\dyck_{n,k}(Y)=1$ \textbf{and} $\textsc{Step3}(S)=\false$}
\State \Return $1$
\Else
\State \Return $0$
\EndIf

\end{algorithmic}
\end{algorithm}

\subsection{Query complexity}\label{sec:algo-compl}
In this section we estimate the query complexity of $\dyck_{k,n,t}$ and discuss properties of Algorithm \ref{alg:dyck-t}.
\begin{theorem}\label{th:upper}
Algorithm \ref{alg:dyck-t} for solving $\dyck_{k,n,t}$, has query complexity $O(\sqrt{n}(\log n)^{0.5 k})$ and a constant two-side bounded error probability $\varepsilon<0.5$. 
\end{theorem}
\begin{proof}
We start with the query complexity of the algorithm.

The complexity of Step 1 is obviously equal to the one of Grover's algorithm, i.e. to $O(\sqrt{n})$.
Lemma \ref{lm:dyck1} estimates the complexity of Step 2 as $O(\sqrt{n}(\log n)^{0.5k})$.
The complexity of Step 3 can be derived from the code listing of Algorithm \ref{alg:step3}: 

$
O(\sum_{v=1}^k \sqrt{n}(\log n)^{0.5(v-1)})=O(\sqrt{n}(\log n)^{0.5(k-1)}).
$

The overall complexity of Algorithm \ref{alg:dyck-t} is

$
O(\sqrt{n})+O(\sqrt{n}(\log n)^{0.5k})+O(\sqrt{n}(\log n)^{0.5(k-1)})
=O(\sqrt{n}(\log n)^{0.5k}).
$

We continue the proof by considering the error probability of the algorithm. Step 1 has error probability at most $0.5$.
Step 2 has constant error probability $\varepsilon_0<0.5$.
Step 3 has error probability at most $1-(1-\varepsilon_1)^k$ for some constant $\varepsilon_1<0.5$.
As each error probability is constant, we can obtain the desired overall error probability $\varepsilon$ by exploiting the technique from \cite{abikkpssv2020}, i.e. by a series of repetitive calls of the algorithm. 
\Endproof
\end{proof}

We finish our discussion with a couple of lower bounds of the query complexity.
\begin{theorem}
There exists a constant $c_1>0$ such that
$
Q(\dyck_{c_1\ell,n,t})=\Omega(2^{\frac{\ell}{2}}\sqrt{n}).
$
\end{theorem}
\begin{proof}
The similar bound holds for $Q(\dyck_{c_1\ell,n,1})$ \cite[Theorem~6]{abikkpssv2020}. By setting $t=1$ we get that $\dyck_{c_1\ell,n,t}$ is at least as hard as $\dyck_{c_1\ell,n,1}$.
\Endproof
\end{proof}

\begin{theorem}
For any $\gamma>0$, there exists a constant $c_2>0$ such that

$
Q(\dyck_{c_2\log n,n,t})=\Omega(n^{1-\gamma}).
$
\end{theorem}
\begin{proof}
The similar bound holds for $Q(\dyck_{c_2\log n,n,1})$ that was presented in \cite[Theorem~5]{abikkpssv2020}. By setting $t=1$ we get that $\dyck_{c_2\log n,n,t}$ is at least as hard as $\dyck_{c_2\log n,n,1}$.
\Endproof
\end{proof}

\section{Generalizing Step 1 Algorithm}\label{sec:maxsearch}
The restriction for all the brackets to be encoded by positive integers up to $2t$, is quite significant for the proposed quantum algorithm. In contrast, formulation of a more natural problem could assume arbitrary encoding of different kinds of brackets. For example, a string could consist of brackets like ``\texttt{(~)}'', ``\texttt{[~]}'', ``\texttt{\{~\}}'' in arbitrary encoding like ASCII, UTF-32, etc. Under these circumstances, one still can distinguish the type of a certain bracket; and still can determine whether a certain bracket is opening or closing; but one cannot anymore determine how many different types of brackets occur in the string.

Formally speaking, the fragment ``at most $t$ types of brackets'' from our definition of $\dyck_{k,n,t}$ means \sloppy\mbox{$\left|\{Type(s_i): 1\leq i\leq n\}\right|\leq t$} rather than \sloppy\mbox{$\max_{i\in\{1,\ldots,n\}} Type(s_i) \leq t$} which was assumed throughout the paper.
Hereafter we refer to such a more general formulation of the problem as $\dyck'_{k,n,t}$.
The implementation of Step 1 from Section \ref{sec:t-types} is not suitable for solving $\dyck'_{k,n,t}$, whereas the rest of the algorithm does not depend on whether the codes of the types of brackets are consecutive or not.

We note that in many cases this won't be an issue, as the number of different types of brackets $t$ typically is a small constant like $2$, $3$ or $4$. However the following problem could be of certain interest even if not connected with $\dyck'_{k,n,t}$:

\begin{problem}\label{problem:arbitrary-encoding}
Given a string $S$ of length $n$, and a limitation parameter $t$, determine whether
\sloppy\mbox{$\left|\left\{Type\left(s_i\right): 1\leq i\leq n\right\}\right|\leq t$}.
\end{problem}

Note that Step 1 from Section \ref{sec:t-types} obviously reduces to Problem \ref{problem:arbitrary-encoding}.

In the rest of this section we propose an algorithm for solving this problem and thus generalize our solution to $\dyck'_{k,n,t}$, i.e. to the case with arbitrarily encoded sequences of brackets.

Let $T$ be an integer such that $2T$ is an upper bound for the code of a bracket in the input string (e.g. the size of the input alphabet). 
Let $Type:\{1,\dots,2T\}\to\{1,\dots,t\}$ be a function that returns the type of a bracket.
Let $q:\{1,\dots,n\}\times\{1,\dots,2T+1\}\to\{0,\dots,2T\}$ be a function which returns
\begin{itemize}
\item $q(i,r) = Type(i)$ \qquad if $Type(i)<r$; or
\item $q(i,r) = 0$ \qquad otherwise.
\end{itemize}

We consider the following procedure: 

\begin{itemize}
    \item[Step $1$] Compute $y_1 = max\{q(i,2T+1), 1\leq i\leq n\}$ by using D\"{u}rr's and H{\o}yer's algorithm for finding the maximum \cite{dh96}. Thus we compute the maximum among all the codes of brackets. 
    \item[Step $2$] Compute $y_2 = max\{q(i,y_1), 1\leq i\leq n\}$ in the same manner, the second-biggest code among all the codes of brackets.
    \item[\ldots] \ldots
    \item[Step $j$] Compute $y_j = max\{q(i,y_{j-1}), 1\leq i\leq n\}$.
\end{itemize}

This procedure lasts until $y_j=0$, which means that there are no bracket codes less than $y_{j-1}$ and that there are exactly $j-1$ different types of brackets contained in string $S$. Then condition $j-1 \leq t$ indicates whether Step 1 was executed correctly.
We formalize this idea in the code listing of Algorithm \ref{alg:step1}, assuming subroutine $\textsc{QMax}(q(*,y_1), 1,n)$ to implement the quantum algorithm for maximum search \cite{dh96}.

\vspace{-0.5cm}
\begin{algorithm}
\caption{Step 1 for solving $\dyck'_{k,n,t}$}\label{alg:step1}
\begin{algorithmic}
\State $j\gets 1$
\State $y_1\gets\textsc{QMax}(q(*,2T+1), 1,n)$
\While{$y_j\neq 0$}
\If{$j>t$}
\State \Return $0$
\EndIf
\State $j\gets j+1$
\State $y_{j}\gets\textsc{QMax}(q(*,y_{j-1}), 1,n)$
\EndWhile
\State \Return $1$
\end{algorithmic}
\end{algorithm}

\vspace{-0.5cm}
\begin{lemma}\label{lm:step1}
The query complexity of Algorithm \ref{alg:step1} is $O(t\sqrt{n}\log t)$, and the error probability is some constant $\varepsilon<1$.
\end{lemma}
\begin{proof}
The expected query complexity of $\textsc{QMax}(q(*,y_{j-1}), 1,n)$
is $O(\sqrt{n})$ \cite{dh96}. According to Markov's inequality, also the exact query complexity of $\textsc{QMax}(q(*,y_{j-1}), 1,n)$ is $O(\sqrt{n})$.
As the error probability of $\textsc{QMax}$ is some constant, repeating it $2\log_2 t$ times results in the error probability $O(\frac{1}{t^2})$.
\Endproof
\end{proof}

If $t=O(\log n ^{0.5(k-1)})$, then the query complexity of Algorithm \ref{alg:step1} (run at Step 1) won't exceed the complexity of Step 2, and the overall complexity of the algorithm will remain the same.

\begin{theorem}
Algorithm \ref{alg:dyck-t} with Step 1 implemented by Algorithm \ref{alg:step1}, solves $\dyck'_{k,n,t}$.
If $t=O(\log n ^{0.5(k-1)})$, then the query complexity of this solution is $O(\sqrt{n}(\log n)^{0.5k})$, and the two-side bounded error probability is $\varepsilon<0.5$.
\end{theorem}
\begin{proof}
According to Lemma \ref{lm:step1}, the query complexity of Step 1 is $O(\sqrt{n}\log n ^{0.5(k-1)}\log\log n)=O(\sqrt{n}\log n ^{0.5k})$. Steps 2 and 3 are the same as in Algorithm \ref{alg:dyck-t}, with complexities  resp. $O(\sqrt{n}\log n ^{0.5k})$ and $O(\sqrt{n}\log n ^{0.5(k-1)})$ proven as for Theorem \ref{th:upper}. Thus the overall query complexity is

$
O(\sqrt{n}\log n ^{0.5k})+O(\sqrt{n}\log n ^{0.5k})+O(\sqrt{n}\log n ^{0.5(k-1)})=O(\sqrt{n}\log n ^{0.5k}).
$

The estimation of the error probability is analogous to the one in the proof of Theorem \ref{th:upper}
\Endproof
\end{proof}

Although we strongly believe that there exists a more efficient quantum algorithm for solving Step 1 of $\dyck'_{k,n,t}$, but for now we limit our considerations with the just proposed iterative maximum search.

{\bf Acknowledgements.} The research is funded by the subsidy allocated to Kazan Federal University for the state assignment in the sphere of scientific activities, project No. 0671-2020-0065.

\bibliographystyle{plain}
\bibliography{tcs}

\end{document}